\documentclass[conference]{IEEEtran}

 \IEEEoverridecommandlockouts
 \usepackage{ifpdf}
 \usepackage{cite}
 \usepackage[pdftex]{graphicx}
 \usepackage{amsmath}
 \usepackage{algorithmic}
 \usepackage{array}
 \usepackage{subfigure}
 \usepackage{caption}
 \usepackage[caption=false,font=footnotesize]{subfig}
 \usepackage{fixltx2e}
 \usepackage{stfloats}
 \usepackage{url}
 \usepackage{amsthm}
 \usepackage{amsmath}
 \usepackage{amsfonts}
 \usepackage{amssymb}
 \usepackage[T1]{fontenc} 
 \usepackage{amsmath}
 \usepackage[cmintegrals]{newtxmath}
 \usepackage{bm} 
 \usepackage[all]{xy}
 \usepackage{enumitem}
 \usepackage{float}
 \usepackage{amsmath}
 \usepackage[usenames, dvipsnames]{color}
 \usepackage{ifpdf}
 \usepackage{cite}
 \usepackage[pdftex]{graphicx}
 \usepackage{amsmath}
 \usepackage{algorithmic}
 \usepackage{array}
 \usepackage{mathtools}
 \DeclarePairedDelimiter{\abs}{\lvert}{\rvert}
 \usepackage{amsmath}
 \usepackage{amsfonts}
 \usepackage{amssymb}
 \usepackage{caption}
 \usepackage[caption=false,font=footnotesize]{subfig}
 \usepackage{fixltx2e}
 \usepackage{stfloats}
 \usepackage{url}
 \usepackage{mathtools}
 \usepackage{amsthm}
 \usepackage{amsmath}
 \usepackage{amsfonts}
 \usepackage{amssymb}
 \usepackage[T1]{fontenc} 
 \usepackage{amsmath}
 \usepackage[cmintegrals]{newtxmath}
 \usepackage{bm} 
 \usepackage[all]{xy}
 \usepackage{enumitem}
 \usepackage{float}
 \usepackage{amsmath}
 \usepackage[dvipsnames]{xcolor}
 \usepackage{multirow}
 \usepackage{subfig}
 \usepackage{multicol}
 \usepackage{graphicx}
 \usepackage[caption=false]{subfig}
 \usepackage{amssymb}
 \usepackage{amsmath}
 \usepackage{commath}
 \usepackage{graphicx,bm}
 \usepackage{verbatim}
 
 \usepackage{diagbox}
 
 \theoremstyle{definition}
 
 \newtheorem{thm}{Theorem}

 \newtheorem{lem}{Lemma}
 
 \newtheorem{define}{Definition}

 \newcommand{\no}{\nonumber}

\begin{document}
	\title{Asymptotic Privacy Loss due to Time Series
Matching of Dependent Users}	
	\author{\IEEEauthorblockN{Nazanin Takbiri*}
		\IEEEauthorblockA{Electrical and\\Computer Engineering\\
			UMass-Amherst\\
			ntakbiri@umass.edu\thanks{Nazanin Takbiri and Minting Chen contributed equally to this work. This work was supported by the National Science Foundation under grants CCF--1421957 and CNS--1739462.}}
		\and
		\IEEEauthorblockN{Minting Chen*}
		\IEEEauthorblockA{Electrical and\\Computer Engineering\\
			UMass-Amherst\\
			mintingchen@umass.edu }
		\and
		\IEEEauthorblockN{Dennis L. Goeckel}
		\IEEEauthorblockA{Electrical and\\Computer Engineering\\
			UMass-Amherst\\
			goeckel@ecs.umass.edu}
		\and
		\IEEEauthorblockN{Amir Houmansadr}
		\IEEEauthorblockA{Information and \\Computer Sciences\\
			UMass-Amherst\\
			amir@cs.umass.edu}
		\and
		\IEEEauthorblockN{Hossein Pishro-Nik}
		\IEEEauthorblockA{Electrical and\\Computer Engineering\\
			UMass-Amherst\\
			pishro@ecs.umass.edu}
	}

\maketitle

\begin{abstract}	
The Internet of Things (IoT) promises to improve user utility by tuning applications to user behavior, but revealing the characteristics of a user's behavior presents a significant privacy risk. Our previous work has established the challenging requirements for anonymization to protect users' privacy in a Bayesian setting in which we assume a powerful adversary who has perfect knowledge of the prior distribution for each user's behavior. However, even sophisticated adversaries do not often have such perfect knowledge; hence, in this paper, we turn our attention to an adversary who must learn user behavior from past data traces of limited length. We also assume there exists dependency between data traces of different users, and the data points of each user are drawn from a normal distribution. Results on the lengths of training sequences and data sequences that result in a loss of user privacy are presented.

\end{abstract}



\section{Introduction}
\label{intro}

The Internet of Things (IoT) allows users to share and access information on a large scale, but the IoT also comes with a significant threat to users' privacy: leakage of sensitive information~\cite{ukil2014iot}. There are two main approaches to augment privacy for IoT users: identity perturbation and data perturbation. Identity perturbation (or anonymization) is the removal of the identifying information from a set of data to protect privacy~\cite{hoh2005protecting,freudiger2007mix}, whereas data perturbation (or obfuscation) is the process of adding noise to the data~\cite{shokri2012protecting}. A cost for employing these Privacy-Protection Mechanisms (PPMs) is a reduction in utility and efficiency of the user data; therefore, optimizing the level of PPMs is of great interest.

In~\cite{Naini2016, bordenabe2014optimal}, a comprehensive analysis of the asymptotic (in the length of the time series) optimal matching of time series to source distributions is presented in a non-Bayesian setting, where the number of users is a fixed, finite value. In contrast, we have adopted a Bayesian setting in
~\cite{tifs2016, Nazanin_IT,takbiri2018privacy, nazanin_ISIT2018}, where a powerful adversary is assumed to have accurate prior distributions for user behavior through past observations or other sources. We consider the length of observations available to the adversary that guarantee privacy, or, conversely, the length of observations for which privacy is compromised~\cite{tifs2016, Nazanin_IT,takbiri2018privacy, Nazanin_CISS2019}. In \cite{takbiri2018privacy}, our most significant results are converse results that demonstrate that this powerful adversary can exploit correlations between the data of different users to compromise user privacy. Thus, a limitation of the converse results of \cite{takbiri2018privacy} is that they are predicated on a very powerful adversary, which, while desirable for (forward) results that guarantee privacy, should be relaxed if possible for (converse) results that demonstrate the loss of privacy. Our main contribution in this letter is to resolve this limitation by developing converse results assuming that the adversary does not have perfect knowledge of the statistics of users' behavior but rather a set of data containing past behavior for each user, from which the adversary can learn user characteristics. 

An initial investigation in~\cite{KeConferance2017} was restricted to obtaining the necessary conditions for breaking privacy for a finite number of users. In contrast, here we turn our attention to this problem in the most general setting of our prior work with an asymptotically large number of users~\cite{tifs2016, Nazanin_IT,takbiri2018privacy, Nazanin_CISS2019}. In particular, contrary to~\cite{KeConferance2017,Nazanin_CISS2019}, we allow for inter-user correlation as in \cite{takbiri2018privacy}. Furthermore, we bring our results closer to practice by, rather than presuming the user's data is discrete-valued \cite{tifs2016, Nazanin_IT,takbiri2018privacy, Nazanin_CISS2019}, considering a Gaussian model for users' data, as Gaussian distributed data has been considered in various domains, e.g., sensor networks~\cite{xiao2005scheme} and distributed consensus~\cite{wagner2004resilient}, as a promising substitute to real data. Cullina et al. have also investigated the related problem of database alignment in a different framework, in which the conditions for exact recovery of the correspondence between database entries have been obtained~\cite{Negar}. 	
The rest of the paper is organized as follows. In Section~\ref{frame}, we present the system model, metrics, and definitions. Then, we present the construction and analysis in Section~\ref{analysis}, and in Section~\ref{conclusion}, we draw conclusions.


\section{Framework}
\label{frame}

Define a system with $n$ users, where each user creates a series of $m$ data points, and the adversary seeks to identify users based on this collection of data points.  Let $X_u(k)$ be the data point of user $u$ at time $k$.  The vectors $\textbf{X}_u$ will be termed the ``actual data set'': $$\textbf{X}_u = [X_u(1), X_u(2),\cdots,X_u(m)]^T, \ \  u \in \{1,2,\cdots, n\}.$$
For every user, there is also a series of $l$ data points representing the user's past behavior; we term these vectors $\textbf{W}_u$ the ``Learning Data Set'':  
$$\textbf{W}_u = [W_u(1), W_u(2),\cdots,W_u(l)]^T, \ \ u \in \{1,2,\cdots, n\}.$$
For $k \in \{1, 2, \cdots, m\}$ and $k' \in \{1, 2, \cdots, l\}$, ${X}_u(k)$ and ${W}_u(k')$ are drawn from a user-specific probability distribution.  In particular, we assume that the points in the data sets of a given user user are drawn from a normal distribution $N(\mu_u,\sigma^2)$, where $\mu_u$ is the mean of the data of user $u$ and $\sigma^2$ is its variance.  While the $\mu_u$'s are unknown to the adversary, each of them is drawn independently from a continuous density function $f_\mu(x)$.  We assume the mild technical condition that there exists $\delta > 0$ such that $f_\mu(x)<\delta$ for all $x$.  Further, the points in the two data sets ${X}_u(k)$ and ${W}_u(k')$ are drawn independently from those in the other set, and, within each set, independently across index ($k$ or $k'$), although there may be inter-user correlation as described below.   

In order to protect the privacy of the users, anonymization is employed as a PPM that conceals the mapping between the learning data set and the actual data set by using a random permutation function $(\Pi)$. The result of permuting $\textbf{X}_u$ yields the ``observed data set'': $$\textbf{Y}_u = [Y_u(1), Y_u(2),\cdots,Y_u(m)]^T,  \ \ u \in \{1,2,\cdots, n\},$$ 
where each ${Y}_u(k)$ has a normal distribution $N(\mu_{\Pi^{-1}(u)},{\sigma^2});$ $\mu_{\Pi^{-1}(u)}$ is the mean of the trace in the actual data set that gets mapped to the $u^{th}$ position in the observed data set by the permutation.  Thus, we have $\textbf{Y}_u = \textbf{X}_{\Pi^{-1}(u)}$ and $\textbf{Y}_{\Pi(u)} = \textbf{X}_u$.  Figure \ref{fig:learning} shows the relation of the three data sets.

\begin{figure}[h]
	\centering
	\includegraphics[width =1\linewidth]{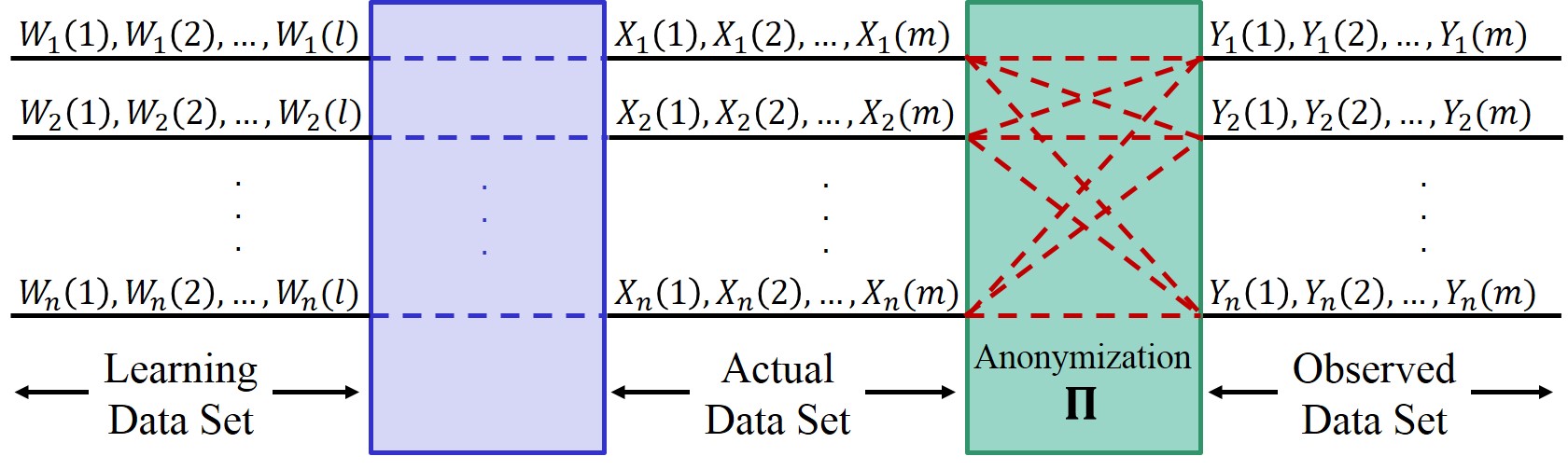}
	\caption{The goal of the adversary: match each sequence in the learning data set $\textbf{W}_u$, $u=1,2,\ldots,n$, to a sequence in the observed data set $\textbf{Y}_u$, $u=1,2,\ldots,n$.}
	\label{fig:learning}
\end{figure}


{\hspace{-0.14 in}\textbf{Association Graph:}} 
The dependencies between users are
modeled by an association graph $G(\mathcal{V}, E)$, where $\mathcal{V}$ represents the nodes and $E$ represents the edges. In this graph, two users are connected if they are dependent. More specifically, we assume 
\begin{itemize}
	\item $(u,u')\in E$ if and only if $Cov_{uu'} > 0$,
	\item $(u,u')\notin E$ if and only if $Cov_{uu'} = 0$,
\end{itemize} 
where $Cov_{uu'}$ is the covariance of the data of user $u$ and user $u'$ at any given time.
\begin{figure}
	\centering
	\includegraphics[width=0.7\linewidth]{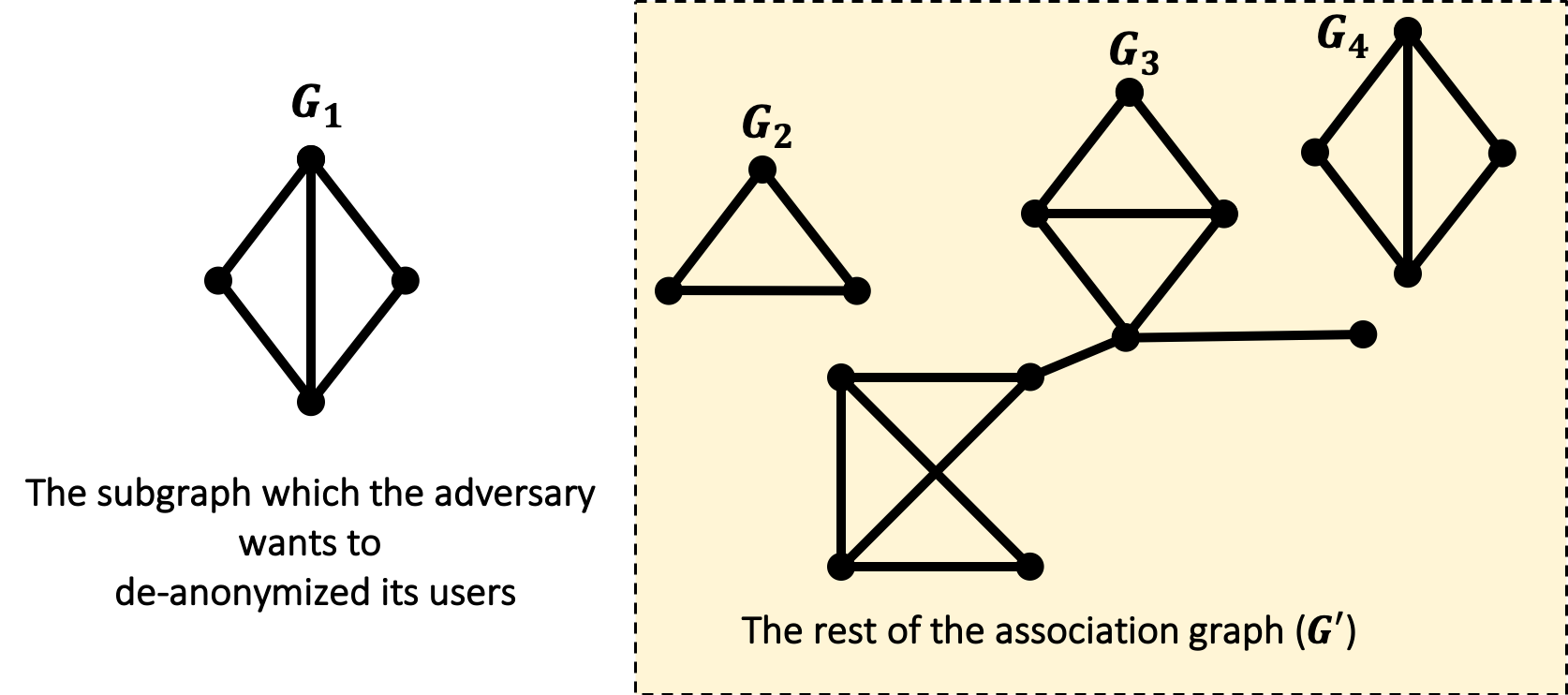}
	\centering
	\caption{The structure of the association graph $(G)$: Group $g$ with $s_g$ vertices is disjoint from the remainder of the association graph $(G')$.}
	\label{fig:graph}
\vspace*{-0.15in}
\end{figure}

\hspace{-0.14 in}\textbf{Adversary Model}: 
The adversary performs statistical matching between the learning data set $\big\{\textbf{W}_u, u = 1,2,\cdots,n \big\}$ and the observed data set $\big\{\textbf{Y}_u, u = 1,2,\cdots,n\big\}$ to match traces in the former, which contains identifying information, with traces in the latter.  We assume the adversary knows the structure of the association graph $G(\mathcal{V}, E)$.  The adversary also has knowledge of the anonymization mechanism (i.e. that a random permutation is employed), but not the realization of the random permutation.

	We define a user having no privacy as~\cite{Nazanin_IT}:
	\begin{define}
		User $u$ has \emph{no privacy} at time $k$ if there exists an algorithm for the adversary to estimate $X_u(k)$ perfectly as $n$ goes to infinity. In other words, as $n \rightarrow \infty$,
		\begin{align}
		\no \forall k\in \mathbb{N}, \ \ \ \mathbb{P}_e(u)\triangleq \mathbb{P}\left(\widehat{X}_u(k) \neq X_u(k)\right)\rightarrow 0,
		\end{align}
		where $\widehat{X}_u(k)$ is the adversary's estimate of $X_u(k)$. 
	\end{define}

\section{Impact Of Employing Training Data On Privacy Using Anonymization}
\label{analysis}


The proof of our key result incorporating learning data sets follows the same three steps as that in \cite{takbiri2018privacy} for the case that the adversary has perfect knowledge of the statistical behavior of the users.  However, because we need to employ learning data and the data points are drawn from Gaussian distributions, there are technical challenges in the second two steps, as illustrated below. 

In the first step, we consider the ability of the adversary to fully reconstruct the structure of the association graph of the anonymized version of the data.
\begin{lem}
	If for any $\lambda>0$, the adversary obtains $m=n^{\lambda}$ points in the observation data set, they can reconstruct $\widetilde{G}=\widetilde{G}(\widetilde{\mathcal{V}}, \widetilde{E})$, where $\widetilde{\mathcal{V}}=\{\Pi(u):u \in \mathcal{V}\}=\mathcal{V}$, such that with high probability, for all $u, u' \in \mathcal{V}$; $ (u,u')\in E$ iff $\left(\Pi(u),\Pi(u')\right)\in \widetilde{E}$. We write this statement as $\mathbb{P}(\widetilde{E}=E)\to 1$.	\label{lemma1}
\end{lem}
\begin{proof}
The reconstruction of the association graph does not require the adversary's knowledge about user statistics (i.e., the values of $\mu_u$'s)~\cite[Lemma 1]{takbiri2018privacy}. Thus, according to the result of ~\cite[Lemma 1]{WCNC2019}, the adversary is able to fully reconstruct the structure of the association graph of the anonymized version of the data with arbitrarily small error probability independent of the length of the learning data set.
\end{proof}
Without loss of generality, assume that User 1 belongs to Group 1 of size $s$.  In contrast to \cite{takbiri2018privacy}: (i) the data points are drawn from a Gaussian distribution; and, more importantly, (ii) the adversary does not know the statistics of the users in Group 1, but rather only has the learning data sets for those users.  In the next step, we demonstrate how the adversary can identify Group $1$ among all of the groups given sufficiently long data traces. 
\begin{lem}
	If for any $\alpha, \alpha' >0$, the adversary obtains learning data sets containing $l=n^{\frac{2}{s}+\alpha'}$ data points of past behavior for each user, and observation data sets containing $m=n^{\frac{2}{s}+\alpha}$ data points for each user, and knows the structure of the association graph, they can identify the traces in the observation data set that correspond to users in Group $1$ with arbitrarily small error probability.
	\label{lemma2_2}
\end{lem}

\begin{proof}
Note that there are at most $\frac{n}{s}$ groups of size $s$ in the system, which we label $1, 2,\cdots,\frac{n}{s}$. 
Define the mean vector for users in Group $1$ as:
$$\mathbf{P}^{(1)} = \left[\mu_1,\mu_2,\cdots,\mu_s\right],$$
and the vectors of empirical averages for the two sets of data which the adversary seeks to match as:
$$\overline{\mathbf{W}}^{(1)} = \left[\overline{W}_1,\overline{W}_2,\cdots,\overline{W}_s\right], \ \ \ \ \ \ \overline{\mathbf{Y}}^{(1)}= \left[\overline{Y}_1,\overline{Y}_2,\cdots,\overline{Y}_s\right],$$ where
$\overline{W}_u = \frac{1}{l} \sum_{i=1}^l W_u(i)$ and $\overline{Y}_u = \frac{1}{m} \sum_{i=1}^m Y_u(i)$.
Let $\Pi_s$ be the set of all permutations on $s$ users; for $\pi_s \in \Pi_s, \pi_s : \left\{1,2,\cdots,s\right\} \rightarrow \left\{1,2,\cdots,s\right\}$ is a one-to-one mapping.  
For any two length-$s$ vectors $\mathbf{U}$ and $\mathbf{V}$, we define a difference function that takes into account any permutation of those vectors:
 $$D\left(\mathbf{U},\mathbf{V} \right) = \min\limits_{\pi_s \in \Pi_s} \left\{||\mathbf{U}-\mathbf{V}_{\pi_s}||_\infty\right\},$$
where $||\mathbf{U}||_{\infty}=\max_{i=1,2,\ldots,k} U_i$ for length-$k$ vector $\mathbf{U}$.  It is straightforward to show that $D\left(\mathbf{U},\mathbf{V} \right)$ satisfies the triangle inequality, which we will employ below.

Now, defining $\mathbf{P}^{(g)}$, $\overline{\mathbf{W}}^{(g)}$, and $\overline{\mathbf{Y}}^{(g)}$ for groups $g=2,3,\ldots,\frac{n}{s}$ in an analogous way to the definitions of $\mathbf{P}^{(1)}$, $\overline{\mathbf{W}}^{(1)}$, and $\overline{\mathbf{Y}}^{(1)}$, respectively, we claim for $m=n^{\frac{2}{s}+\alpha}$, $l=n^{\frac{2}{s}+\alpha'}$, and as 
$n \to \infty$:
\begin{enumerate}
	\item $\mathbb{P}\left(D\left(\overline{\mathbf{W}}^{(1)},\overline{\mathbf{Y}}^{(1)} \right)\leq \Delta_n\right) \rightarrow 1,$
	\item $\mathbb{P}\left( \bigcup\limits_{g=2}^{\frac{n}{s}}D\left(\overline{\mathbf{W}}^{(1)},\overline{\mathbf{Y}}^{(g)} \right)\leq \Delta_n\right) \rightarrow 0,$
\end{enumerate}
where $\Delta_n=n^{-\frac{1}{s}-\frac{\alpha''}{4}}$, and $\alpha''=\min\{\alpha, \alpha'\}$.
For each $u \in\{1, 2, \cdots, n\}$,
		\begin{align}
		\no \mathbb{P}\left(\ \abs{\overline{X}_u-\overline{W}_u}\geq\Delta_n\right)&=\mathbb{P}\left(\ \abs{\left(\overline{X}_u-\mu_u\right)-\left(\overline{W}_u-\mu_u\right)}\geq\Delta_n\right)\\
		\no &\hspace{-0.55 in} \leq \mathbb{P}\left(\ \abs{\overline{X}_u-\mu_u}+\ \abs{\overline{W}_u-\mu_u}\geq\Delta_n\right)\\
		\no &\hspace{-0.55 in} \leq \mathbb{P}\left(\left\{\ \abs{\overline{X}_u-\mu_u}\geq\frac{\Delta_n}{2}\right\} \bigcup  \left\{\ \abs{\overline{W}_u-\mu_u}\geq\frac{\Delta_n}{2}\right\}\right)\\
\no &\hspace{-0.55 in} \leq \mathbb{P}\left(\ \abs{\overline{X}_u-\mu_u}\geq\frac{\Delta_n}{2}\right)+\mathbb{P}\left(\ \abs{\overline{W}_u-\mu_u}\geq\frac{\Delta_n}{2}\right)\\
 	&\hspace{-0.55 in} \leq e^{\frac{-m\Delta_n^2}{8\sigma^2}} + e^{\frac{-l\Delta_n^2}{8\sigma^2}} \leq 2e^{-\frac{n^{\frac{\alpha''}{2}}}{8\sigma^2}},
		\label{eqYW}
		\end{align}
where $\alpha''=\min \{\alpha, \alpha'\}$. The first inequality follows from the triangle inequality. 
The union bound yields the third inequality, and the fourth inequality is based on the error function inequality $\mbox{erf}(x)\geq1-e^{-x^2}$. By employing (\ref{eqYW}) and applying the union bound for all of the users in a group with size $s$, we have for any group $g$ that: 
\begin{align}
\no \mathbb{P}\left(D\left(\overline{\mathbf{W}}^{(g)},\overline{\mathbf{Y}}^{(g)} \right)\geq \Delta_n\right) & \leq \sum_{u=1}^{s} \mathbb{P}\left(\ \abs{\overline{X}_u-\overline{W}_u}\geq\Delta_n\right)\\
\no &= s\mathbb{P}\left(\ \abs{\overline{X}_u-\overline{W}_u}\geq\Delta_n\right)\\
& \leq 2se^{-\frac{n^{\frac{\alpha''}{2}}}{8\sigma^2}}.
\label{delta}
\end{align}
Hence, letting $g=1$, $\mathbb{P}\left(D\left(\overline{\mathbf{W}}^{(1)},\overline{\mathbf{Y}}^{(1)} \right)\leq \Delta_n\right) \rightarrow 1$, as $n \to \infty.$  Next, we want to show that $\mathbb{P}\left( \bigcup\limits_{g=2}^{\frac{n}{s}}D\left(\overline{\mathbf{W}}^{(1)},\overline{\mathbf{Y}}^{(g)} \right)\leq \Delta_n\right) \rightarrow 0$, as $n \to \infty$.  We do this in three steps.
\begin{itemize}[wide=3pt]
\item First, recalling the (mild) technical condition that the probability density function from which the user means is drawn is upper bounded by $\delta$ and that the user means are drawn independently, for Group $g$ we obtain:
		\begin{align}
		\no \mathbb{P}\left( 
||{\mathbf{P}^{(1)}}-{\mathbf{P}^{(g)}}||_\infty 
\leq 4\Delta_n\right) & \leq (8\Delta_n)^s\delta
		 = 8^sn^{-1-\frac{s\alpha''}{4}}\delta.
		\end{align}
		Similarly, for all $\pi_s \in \Pi_s$, we have 
		\begin{align}
		\no \mathbb{P}\left(
||{\mathbf{P}^{(1)}}-{\mathbf{P}^{(g)}}_{\pi_s}||_\infty 
\leq 4\Delta_n\right) \leq 8^sn^{-1-\frac{s\alpha''}{4}}\delta.
		\end{align}
		Employing union bounds, since $|\Pi_s|=s!$, we have
		\begin{align}
		\no & \mathbb{P}\left( \bigcup\limits_{g=2}^{\frac{n}{s}} \left\{D\left({\mathbf{P}^{(g)}_{\pi_s}},{\mathbf{P}^{(1)}} \right)\leq 4\Delta_n\right\} \right)\\
		\no & \hspace{0.5 in} = \mathbb{P}\left( \bigcup\limits_{g=2}^{\frac{n}{s}} \left\{ \bigcup\limits_{\pi_s \in \Pi_s} 
\left \{ ||{\mathbf{P}^{(1)}}-{\mathbf{P}^{(g)}}_{\pi_s}||_\infty \leq 4\Delta_n\right\} 
\right\} \right)\\
		\no &\hspace{0.5 in} \leq \sum_{g=2}^{\frac{n}{s}}\sum_{\pi_s \in \Pi_s}\mathbb{P} \left(
||{\mathbf{P}^{(1)}}-{\mathbf{P}^{(g)}}_{\pi_s}||_\infty \leq 4\Delta_n
\right)\\
		\no &\hspace{0.5 in} \leq \frac{n}{s}s!8^sn^{-1-\frac{s\alpha''}{4}}\delta = (s-1)!8^sn^{-\frac{s\alpha''}{4}}\delta\rightarrow 0,
		\end{align}  
		as $n \to \infty$.  Thus, with high probability, the difference between all $\mathbf{P}^{(g)}$, $g \geq 2$, and $\mathbf{P}^{(1)}$ is bigger than $4 \Delta_n$. 

\item Second, for all $u \in \{2,3,\cdots, n\}$, $\mbox{erf}(x)\geq1-e^{-x^2}$ yields 
		\begin{align}
		\no \mathbb{P}\left(\ \abs{\overline{W}_u-\mu_u}\geq \Delta_n \right) &\leq e^{-\frac{l\Delta_n^2}{2\sigma^2}} \leq e^{-\frac{n^{\frac{\alpha''}{2}}}{2\sigma^2}}.
		\end{align}
		Thus, by employing union bounds, we have
		\begin{align}
		\no \mathbb{P}\left( D\left(\overline{\mathbf{W}}^{(g)},\mathbf{P}^{(g)}\right)\geq \Delta_n \right) 
		& \leq \mathbb{P}\left(\lVert \overline{\mathbf{W}}^{(g)}-{\mathbf{P}^{(g)}} \rVert_\infty \geq \Delta_n \right)  \\
&\leq \sum_{u \in \text{Group } l} \mathbb{P}\left(\ \abs{\overline{W}_u-\mu_u}\geq \Delta_n \right) \nonumber \\
		\no & = se^{-\frac{n^{\frac{\alpha''}{2}}}{2\sigma^2}}. 
		\end{align}
		Now, for $g=1$, as $n \to \infty$, we have
		$$\mathbb{P}\left( D\left(\overline{\mathbf{W}}^{(1)},\mathbf{P}^{(1)}\right)\geq \Delta_n \right) \leq se^{-\frac{n^{\frac{\alpha''}{2}}}{2\sigma^2}} \rightarrow 0.$$

\item Thirdly, since we have shown above that with high probability, $D\left(\mathbf{P}^{(g)},\mathbf{P}^{(1)}\right)\geq 4\Delta_n$ and $D\left(\overline{\mathbf{W}}^{(g)},\mathbf{P}^{(g)} \right)\leq \Delta_n$, for all $l \in \{2,3,\cdots, \frac{n}{s}\}$, by the triangle inequality we have 
		\begin{align}
		\no \mathbb{P}\left(\ D\left(\overline{\mathbf{W}}^{(g)},\overline{\mathbf{W}}^{(1)} \right) \leq 2\Delta_n\right) & \leq \mathbb{P}\left(\ D\left(\overline{\mathbf{W}}^{(g)},\mathbf{P}^{(g)} \right) \geq\Delta_n\right)\\
		\no & \leq se^{-\frac{n^{\frac{\alpha''}{2}}}{2\sigma^2}},\ \
		\end{align}
		and by applying a union bound, as $n \to \infty$, 
		\begin{align}
		\no & \mathbb{P}\left( \bigcup\limits_{g=2}^{\frac{n}{s}} \left\{ D\left(\overline{\mathbf{W}}^{(g)},\overline{\mathbf{W}}^{(1)} \right)\leq 2\Delta_n\right\}\right)\\
		\no &\hspace{1 in} \leq \sum_{g=2}^{\frac{n}{s}} \mathbb{P}\left(\ D\left(\overline{\mathbf{W}}^{(g)},\overline{\mathbf{Y}}^{(1)}\right)\leq 2\Delta_n\right)\\
		\no &\hspace{1 in} =\frac{n}{s}se^{-\frac{n^{\frac{\alpha''}{2}}}{2\sigma^2}} = ne^{-\frac{n^{\frac{\alpha''}{2}}}{2\sigma^2}} \rightarrow 0.
		\end{align}	
\item Finally, since we have shown that, with high probability, $D\left(\overline{\mathbf{W}}^{(g)},\overline{\mathbf{Y}}^{(g)} \right)\leq \Delta_n$ and $D\left(\overline{\mathbf{W}}^{(g)},\overline{\mathbf{W}}^{(1)} \right)\geq2\Delta_n$, for all $g \in \{2,3,\cdots, \frac{n}{s}\}$: 
		\begin{align}
		\no \mathbb{P}\left(\ D\left(\overline{\mathbf{W}}^{(1)},\overline{\mathbf{Y}}^{(g)}\right)\leq \Delta_n\right) & \leq \mathbb{P}\left(\ D\left(\overline{\mathbf{Y}}^{(g)},\overline{\mathbf{W}}^{(g)}\right)\geq\Delta_n\right)\\
		&\leq 2se^{-\frac{n^{\frac{\alpha''}{2}}}{8\sigma^2}} \rightarrow 0,
		\label{eq10}
		\end{align}  
		as $n \to \infty$, where the second inequality follows from ($\ref{delta}$). 
Employing (\ref{eq10}) and a union bound, as $n \to \infty$ we have 
		\begin{align}
		\no & \mathbb{P}\left(\bigcup\limits_{g=2}^{\frac{n}{s}} \left\{ D\left(\overline{\mathbf{W}}^{(1)},\overline{\mathbf{Y}}^{(g)}\right) \leq \Delta_n \right\}\right) \\
		\no &\hspace{1 in} \leq \sum_{g=2}^{\frac{n}{s}} \mathbb{P}\left(\ D\left(\overline{\mathbf{Y}}^{(g)},\overline{\mathbf{W}}^{(g)}\right)\geq\Delta_n\right)\\
		\no &\hspace{1 in} \leq \frac{n}{s}2se^{-\frac{n^{\frac{\alpha''}{2}}}{8\sigma^2}} = 2ne^{-\frac{n^{\frac{\alpha''}{2}}}{8\sigma^2}} \rightarrow 0.
		\end{align}
\end{itemize}
Hence, we can conclude that if $m=n^{\frac{2}{s}+\alpha}$, $l=n^{\frac{2}{s}+\alpha'}$, and $n \to \infty$, the adversary can identify the data traces in the observed data set belonging to users in Group 1 with small error probability.  \end{proof}

Finally, in the following lemma, we show that once the data traces in the observed data set belonging to users in Group 1 are identified, the adversary can identify the data trace in the observed data set for each of the members of Group $1$ with arbitrarily small error probability.

\begin{lem}
	If for any $\alpha, \alpha' >0$, the adversary obtains learning data sets containing $l=n^{\frac{2}{s}+\alpha'}$ data points of past behavior for each user, and observation data sets containing $m=n^{\frac{2}{s}+\alpha}$ data points for each user, and knows which traces in the observation data set belong to members of Group $1$, the adversary can identify the trace in the observation set belonging to user 1 with arbitrarily small error probability.
	\label{lemma3_2}
\end{lem}

\begin{proof}
We claim that, for $m=n^{\frac{2}{s}+\alpha}$, $l=n^{\frac{2}{s}+\alpha'}$, and as $n \to \infty$,
\begin{enumerate}
	\item $\mathbb{P}\left(\ \abs{\overline{X}_1-\overline{W}_1}\leq \Delta_n\right) \rightarrow 1$,
	\item $\mathbb{P}\left(\bigcup\limits_{u=2}^{s} \abs{ \overline{X}_u-\overline{W}_1}\leq \Delta_n\right) \rightarrow 0$,
\end{enumerate}
where $\Delta_n= n^{-(\frac{1}{s}+\frac{\alpha''}{4})}$, and $\alpha''=\min\{\alpha, \alpha'\}$.



\begin{enumerate}[wide=0pt]

\item The first claim follows from (\ref{eqYW}) with $u=1$ by taking $n \rightarrow \infty$.

\item Next we establish the second claim.  
Recall the (mild) technical assumption that $f_{\mu}(x) < \delta$ for some $\delta$. Then, for all $u \in \{2,3,\cdots, n\}$, $$\mathbb{P}(|\mu_u-\mu_1|\leq 4\Delta_n) \leq8\Delta_n\delta.$$
A union bound yields
\begin{align}
\no \mathbb{P}\left( \bigcup\limits_{u=2}^{s} \left\{ \ \abs{\mu_u-\mu_1} \leq 4\Delta_n \right \} \right) & \leq \sum_{u=2}^{s} \mathbb{P} \left(\ \abs{\mu_u-\mu_1} \leq 4\Delta_n \right)\\
\no & \leq8s\Delta_n\delta =8sn^{-1-\frac{\alpha''}{4}}\delta \rightarrow 0, 
\end{align}  
as $n \to \infty$. This means that, with high probability, all of the $\mu_u$ for $u > 1$ fall outside of the range of $\mu_1 \pm 4\Delta_n$.

Next, for all $u \in \{2,3,\cdots, n\}$, $\mbox{erf}(x)\geq1-e^{-x^2}$ yields $$\mathbb{P}\left(\ \abs{\overline{W}_u-\mu_u}\geq\Delta_n\right)\leq e^{-\frac{l\Delta_n^2}{2\sigma^2}}.$$ Thus, for $u=1$, as $n \to \infty$, we have $$\mathbb{P}\left(\ \abs{\overline{W}_1-\mu_1} \geq\Delta_n\right)\leq e^{-\frac{n^{\frac{\alpha''}{2}}}{2\sigma^2}}\rightarrow 0,$$ which means $\overline{W}_1$ is inside $\mu_1\pm\Delta_n$ with high probability.
 
Thus, we have now shown with high probability that $|\mu_u-\mu_1|\geq4\Delta_n$ and $|\overline{W}_u-\mu_u|\leq \Delta_n$, for all $u \in \{2,3,\cdots, n\}$; thus, the triangle inequality yields:
\begin{align}
\no \mathbb{P}\left(\ \abs{\overline{W}_u-\overline{W}_1} \leq 2\Delta_n\right) & \leq \mathbb{P}\left(\ \abs{\overline{W}_u-\mu_u} \geq\Delta_n\right)\\
\no & \leq e^{-\frac{l\Delta_n^2}{2\sigma^2}} \leq e^{-\frac{n^{\frac{\alpha''}{2}}}{2\sigma^2}}.
\end{align}
Applying a union bound, as $n \to \infty$, we have 
\begin{align}
\no  \mathbb{P}\left(\bigcup\limits_{u=2}^{s} \left\{ \ \abs{\overline{W}_u-\overline{W}_1}\leq 2\Delta_n\right\}\right) &\leq \sum_{u=2}^{s} \mathbb{P}\left(\ \abs{\overline{W}_u-\overline{W}_1}\leq 2\Delta_n\right)\\
\no &= se^{-\frac{n^{\frac{\alpha''}{2}}}{2\sigma^2}} \rightarrow 0.
\end{align}
Finally, since $|\overline{X}_u-\overline{W}_u|\leq \Delta_n$ and $|\overline{W}_u-\overline{W}_1|\geq2\Delta_n$, for all $u \in \{2,3,\cdots, s\}$, with high probability, we can employ (\ref{eqYW}) to obtain: 
\begin{align}
\no \mathbb{P}\left(\ \abs{\overline{X}_u-\overline{W}_1}\leq \Delta_n\right) &= \mathbb{P}\left(\ \abs{\overline{X}_u-\overline{W}_u}\geq\Delta_n\right)\leq 2e^{-\frac{n^{\frac{\alpha''}{2}}}{8\sigma^2}}.
\end{align}  
As $n \rightarrow \infty$, a union bound yields:
\begin{align}
\no \mathbb{P}\left(\bigcup\limits_{u=2}^{s} \left\{ \ \abs{\overline{X}_u-\overline{W}_1} \leq \Delta_n \right\}\right)& \leq \sum_{u=2}^{s} \mathbb{P}\left(\ \abs{\overline{X}_u-\overline{W}_u}\geq\Delta_n\right)\\
\no &\leq 2se^{-\frac{n^{\frac{\alpha''}{2}}}{8\sigma^2}} \rightarrow 0.
\end{align}

\end{enumerate}
\vspace{-10 pt}
\end{proof}

\vspace{-0.15 pt}

From Lemmas~\ref{lemma1},~\ref{lemma2_2}, and~\ref{lemma3_2}, we can conclude that a user will have no privacy if the number of data points ($m$) per user in the observation data set and the number of data points ($l$) per user in the the learning data set are both significantly larger than $n^{\frac{2}{s}}$ as the number of users in the network $(n)$ goes to infinity, and the size of the group to which the user of interest belongs is equal to $s$.
\begin{thm}
\label{thm_main}
For the system model with Gaussian data points of Section~\ref{frame}, where $\textbf{Y}$ is the anonymized version of $\textbf{X}$, and $\textbf{W}$ is the behavioral history of users, user 1 has no privacy at time $k$ if:
\begin{itemize}
	\item The adversary knows the structure of the association graph;
	\item The adversary has access to a $l-$length behavioral history for each of the users, where $l=\Omega\left(cn^{\frac{2}{s}+\alpha'}\right)$ for any $\alpha > 0$; 
	\item The adversary has access to a $m-$length observation for each of the users, where $m=\Omega\left(cn^{\frac{2}{s}+\alpha}\right)$ for any $\alpha > 0$; 
\end{itemize}
\end{thm}

The argument for the case where the adversary has perfect prior knowledge about users' past behavior in the Gaussian case, which is not covered by our prior work, follows from arguments similar to those leading to Theorem~\ref{thm_main} and~\cite[Theorem 1]{WCNC2019}. 

\begin{thm}\label{IndependentCase}
For the system model with Gaussian data points of Section~\ref{frame} where $\textbf{Y}$ is the anonymized version of $\textbf{X}$, user 1 has no privacy at time $k$ if:	\begin{itemize}
		\item The adversary knows the structure of the association graph;
		\item The adversary has access to perfect prior knowledge about users' behavior; 
		\item The adversary has access to a $m-$length observations for each of the users, where $m=\Omega\left(cn^{\frac{2}{s}+\alpha'}\right)$ for any $\alpha' > 0$; 
	\end{itemize}
\end{thm}






\section{Conclusion}
\label{conclusion}

IoT devices provide significant convenience for users, but they can allow an adversary to obtain a user's sensitive information.  In this paper, given that anonymization is employed to ensure users' privacy, we consider the broadest set of assumptions compared to previous work: (i) we assume the adversary only has access to limited data sets for users' past behavior rather than perfect knowledge of the statistics of users' behavior; (ii) we assume the data traces of different users are dependent; (iii) we assume that users' data sequences are governed by an i.i.d.\ Gaussian model.  We established sufficient conditions for an adversary to reconstruct the association graph that represents the dependency between users, identify a specific group of dependent users, and determine all of the members of the identified group, hence breaking the privacy of individual users.  In particular, if the length ($l$) of the learning data set and the length ($m$) of the observed data set are each significantly larger than $n^{\frac{2}{s}}$, users have no privacy.

\bibliographystyle{IEEEtran}
\bibliography{REF}


\end{document}